\documentclass[11pt]{article}

\usepackage{graphicx}
\usepackage{pstricks}
\usepackage{pst-node}
\usepackage{pst-coil}
\usepackage{amsmath}
\usepackage{stmaryrd}
\usepackage{amsthm}
\usepackage{amsmath}
\usepackage{amssymb}
\usepackage{amsfonts}
\usepackage{epsfig}
\usepackage{verbatim}
\usepackage{bm}
\usepackage{newtype}

\newtheorem{theorem}{Theorem}[section]
\newtheorem{lemma}[theorem]{Lemma}
\newtheorem{corollary}[theorem]{Corollary}
\newtheorem{conjecture}[theorem]{Conjecture}

\theoremstyle{definition}
\newtheorem{remark}[theorem]{Remark}
\newtheorem{example}[theorem]{Example}
\newtheorem{definition}[theorem]{Definition}

\makeatletter
\@addtoreset{claim}{theorem}
\makeatother

\newtheorem{observation}{Observation}
\makeatletter
\@addtoreset{observation}{theorem}
\makeatother

\newcommand{\ZN}{\ensuremath{\mathbb{Z}}}
\newcommand{\QN}{\ensuremath{\mathbb{Q}}}
\newcommand{\RN}{\ensuremath{\mathbb{R}}}
\newcommand{\CN}{\ensuremath{\mathbb{C}}}

\newcommand{\fd}[1]{\ensuremath{\mathbb{#1}}}

\newtype{\class}{\mathbf}
\class{P}
\class{NP}
\class{AM}
\class{LOGNP}
\class{coNP}
\class{NTIME}
\class{PSPACE}
\class{EXP}
\class{EXPSPACE}
\class{NEXP}
\class{RP}
\class{coRP}
\class{TA}
\class{RA}
\class{PP}
\class{PH}
\class{CH}
\class{FIXP}
\class[FIXPd]{FIXP_d}
\class[sharpP]{\#P}

\newcommand{\sh}[1]{\class<\Sigma^p_{#1}>}

\newcommand{\NPZ}{\class<\ensuremath{\bm{\exists \ZN}}>}
\newcommand{\NPQ}{\class<\ensuremath{\bm{\exists \QN}}>}
\newcommand{\NPR}{\class<\ensuremath{\bm{\exists \RN}}>}
\newcommand{\NPC}{\class<\ensuremath{\bm{\exists \CN}}>}

\newcommand{\K}{\ensuremath{\emptyset'}}

\DeclareMathOperator{\rk}{\mathrm{rank}}
\DeclareMathOperator{\GF}{\mathrm{GF}}
\DeclareMathOperator{\mrk}{\mathrm{minrank}}

\DeclareMathOperator{\diag}{\mathit{diag}}
\DeclareMathOperator{\ETh}{\mathrm{ETh}}

\newcommand{\HTN}{\textsf{\upshape H$_2$N}}

\def\itemi{\item [$(i)$]}
\def\itemii{\item [$(ii)$]}

\newcommand{\llb}{(}
\newcommand{\rrb}{)}

\psset{unit=1pt}

\begin{document}

\title{The Complexity of Tensor Rank}
\author{
{Marcus Schaefer
} \\
{\small School of Computing} \\[-0.13cm]
{\small DePaul University} \\[-0.13cm]
{\small Chicago, Illinois 60604, USA} \\[-0.13cm]
{\small \tt mschaefer@cdm.depaul.edu}\\[-0.13cm]
\and
{Daniel \v{S}tefankovi\v{c}
} \\
{\small Computer Science Department} \\[-0.13cm]
{\small University of Rochester} \\[-0.13cm]
{\small Rochester, NY 14627-0226} \\[-0.13cm]
{\small \tt stefanko@cs.rochester.edu}\\[-0.13cm]
}

\date{}

\maketitle

\begin{abstract}
 We show that determining the rank of a tensor over a field has the same complexity as deciding the existential theory of that field. This implies earlier \NP-hardness results by H{\aa}stad~\cite{H90}. The hardness proof also implies an algebraic universality result.
\end{abstract}

%to do:
%universality
%mention Z and PIDs

\section{Introduction}

As computer scientists we can think of tensors as
multi-dimensional arrays; $2$-dimensional tensors correspond to
(traditional) matrices, and a $3$-dimensional tensor can be
written as $T = (t_{i,j,k}) \in \fd{F}^{d_1 \times d_2 \times
d_3}$. We will work over various fields, including $\QN$, $\RN$,
and $\CN$, as well as $\GF_p$. The rank of a matrix $M$ (over some
field $\fd{F}$) can be defined as the smallest $k$ so that $M$ is
the sum of $k$ matrices of rank $1$, where a matrix of rank $1$ is
a matrix that can be written as $x \otimes y$, where $x$ and $y$
are one-dimensional vectors (over $\fd{F}$), and $\otimes$ is the
Kronecker (tensor, outer) product. The rank of a tensor can be
defined similarly: a $3$-dimensional tensor $T$ has {\em (tensor)
rank} at most $k$ (over $\fd{F}$) if it is the sum of at most $k$
rank-$1$ tensors, where a rank-$1$ tensor is a tensor of the form
$x \otimes y \otimes z$ (over $\fd{F}$).

H{\aa}stad~\cite{H90} showed that determining the tensor rank over
$\QN$ is an \NP-hard problem; as Hillar and Lim~\cite{HL13} point
out, his proof can be (mildly) adjusted to yield that the tensor
rank problem remains \NP-hard over \RN\ and \CN; this is not
immediate, since, tensor rank can vary depending on the underlying
field (this is a well known fact; we will also see an example
later on). This may suggest that tensor rank problems are equally
intractable. Our goal in this paper is to show that this is not
the case, and that the complexity of the tensor rank problem
ranges wildly, as we consider different underlying fields.

For a field $\fd{F}$, let $\ETh(\fd{F})$ be the set of true
existential first-order statements over $\fd{F}$, sometimes known
as the {\em existential theory of $\fd{F}$}. For example, letting
$\varphi(c) := (\exists x)[x^2 = c]$, we have that $\varphi(2)
\not\in \ETh(\QN)$,  but $\varphi(2) \in \ETh(\RN), \ETh(\CN)$,
and $\varphi(-1) \not\in \ETh(\QN), \ETh(\RN)$, and $\varphi(-1)
\in \ETh(\CN)$. Our main result is that the tensor rank problem
over $\fd{F}$ is polynomial-time equivalent to the existential
theory of $\fd{F}$.

\begin{theorem}\label{thm:trans}
 Let $\fd{F}$ be a field. Given a statement $\varphi$ in $\ETh(\fd{F})$, the existential theory of $\fd{F}$, we can in polynomial time construct a tensor $T_\varphi$ and an integer $k$ so that $\varphi$ is true over $\fd{F}$ if and only if $T$ has tensor rank at most $k$ over $\fd{F}$.
\end{theorem}

The existential theory of any finite field is \NP-complete, so
Theorem~\ref{thm:trans} implies H{\aa}stad's result that the
tensor rank problem is \NP-complete over finite fields~\cite{H90}.
If we use $\NPQ$, $\NPR$, and $\NPC$ for the computational
complexity class associated with deciding $\ETh(\QN)$,
$\ETh(\RN)$, and $\ETh(\CN)$, respectively, then we can rephrase
Theorem~\ref{thm:trans} as saying that the tensor rank problem is
$\NPQ$-complete over the rationals, $\NPR$-complete  over the
reals, and $\NPC$-complete over the complex numbers.\footnote{The
complexity class \NPR\ was introduced explicitly
in~\cite{S09b,SS15} and some other papers, but other researchers
probably thought of $\ETh(\RN)$ as a complexity class before,
e.g.,\ Shor~\cite{S91}, and Buss, Frandsen and
Shallit~\cite{BFS99}.}

While none of these complexity classes have been placed exactly
with respect to traditional complexity classes, we do know that
\[\NP \subseteq \NPC \subseteq \NPR \subseteq \PSPACE.\]
The lower bound is folklore~\cite[Proposition 8]{BFS99}.\footnote{We are
not aware of any stronger lower
bounds on $\exists \fd{F}$ for {\em any} field $\fd{F}$. If we
allow rings, then \NPZ, for example, is undecidable,
its complexity equivalent to the halting problem $\K$.
This was shown in a famous series of results by Davis, Robinson, and
Matiyasevic~\cite{M70,DMR76}.}
The inclusion $\NPC \subseteq \NPR$ follows
from the standard encoding of complex numbers as pairs of reals,
and the upper bound of \PSPACE\ on \NPR\ is due to Canny~\cite{C88}.

\NPR\ appears to contain problems harder than
problems in $\NP$ or $\NPC$: even a---seemingly simple---special
problem in $\NPR$ such as the sum of square roots problem has not
been located in the polynomial-time hierarchy
(see~\cite{ABKDM06}). On the other hand, Koiran~\cite{K96} showed
that $\NPC \subseteq \AM$, where $\AM$ is the class of
Arthur-Merlin games, which is known to lie in $\sh{2}$, the second
level of the polynomial-time hierarchy.\footnote{Koiran's result
assumes the generalized Riemann hypothesis (GRH); as far as we
know there is no unconditional upper bound on $\NPC$ other than
\PSPACE.} This suggests that the tensor rank problem over $\CN$
may be significantly easier to solve (if still hard) than the
tensor rank problem over $\RN$.

The complexity of $\NPQ$ is open, it is not even known (or expected) to be
decidable. The currently best result in that direction is the
undecidability of the $\exists\forall$-theory of $\QN$, using
definability results for $\ZN$ over $\QN$ in the footsteps of Julia
Robinson~\cite{K10,P09}. Any decidability results for the tensor
rank problem over $\QN$ would, by our reduction, imply rather
surprising decidability results for
$\NPQ$.\footnote{If $\ZN$ had an existential
definition in $\QN$, then it would follow that $\NPQ \equiv \NPZ
\equiv \K$. Koenigmsann~\cite{K10} gives some evidence that there
is no such definition (implying that his universal definition of
$\ZN$ in $\QN$ is optimal), however, there may be other routes
towards the undecidability of $\NPQ$, and it may be undecidable
without being as hard as $\K$.}
We do know,
however, that $\NPR \subseteq \NPQ$, since deciding the feasibility of
a set of strict polynomial inequalities is hard for \NPR~\cite{SS15},
and lies in \NPQ.

Figure~\ref{fig:hard} summarizes our results for various fields. We note in particular that the upper bounds imply that there are (at least in principle) algorithms for solving the tensor rank problem over finite fields, $\RN$ and $\CN$.

\begin{figure}[htb]
\begin{tabular}{l|l|l|l}
$\fd{F}$  &  complexity of tensor rank over $\fd{F}$ & lower bound & upper bound\\\hline
$\GF_p$ & \NP-complete~\cite{H90} &&  \\
\CN & \NPC-complete & \NP~\cite{H90,HL13} & $\AM \subseteq \sh{2}$~\cite{K96} \\
\RN & \NPR-complete & \NP~\cite{H90,HL13} & \PSPACE~\cite{C88} \\
\QN & \NPQ-complete & \NPR-hard & $\K$ \\
\end{tabular}
\caption{Complexity of the tensor rank problem over various rings. Previously all these problems were known to be \NP-hard using H{\aa}stad's argument~\cite{H90,HL13}.}\label{fig:hard}
\end{figure}

There are many computational problems related to tensors, and, as Hillar and Lim~\cite{HL13} showed compellingly, most of them are hard. Many of their hardness results are \NP-hardness proofs via direct reductions from \NP-complete problems, however, in one or two cases, they reduce from an \NPR-complete problem, and in those cases they also get \NPR-completeness results (even though they do not state this explicitly); in particular, testing whether $0$ is an eigenvalue of a given tensor over $\RN$ is \NPR-complete (see Example~\ref{ex:eig0} for a correction of their proof).

Our point is that it is important to capture the computational complexity of these algebraic problems more precisely than saying that they are \NP-hard, since there may be a significant variance in their hardness (from \NPC, close to \NP, to \NPR, probably closer to \PSPACE, to \NPQ, likely undecidable). For \NPR, there already is a sizable number of complete problems, starting with Mn\"{e}v's universality theorem showing that stretchability of pseudoline arrangements is complete for \NPR~\cite{M88,S91,RG95}, but also including the rectilinear crossing number~\cite{B91}, segment intersection graphs~\cite{KM94} and many others. Less is known about \NPQ, and \NPC.

Our proof of Theorem~\ref{thm:trans} will work via a minimum rank problem for matrices with multilinear entries; versions of this problem were previously studied by Buss, Frandsen and Shallit~\cite{BFS99}. We also show that both the minimum rank problem and the tensor rank problem exhibit algebraic universality.  Algebraic universality implies that solutions to a problem may require algebraic numbers of high complexity.

\begin{remark}
Shitov~\cite{S16} has recently shown a stronger result---the complexity of the tensor rank
over an {\em integral domain} is the same as the complexity of the existential theory of
that integral domain.
\end{remark}

\section{Definitions and Tools}

\subsection{Tensors}

A (3-dimensional, rational) {\em tensor} is a an array $T = (
t_{ijk} )^{d_1, d_2, d_3}_{i,j,k=1} \in \QN^{d_1 \times d_2 \times
d_3}$. Lower dimensional subarrays of a tensor are known as {\em
fibres} (one dimension) and {\em slices} (two dimensions). We
denote subarrays by using ``$:$'' instead of a variable, e.g.,\
$t_{:jk}$ is a column-fibre of $T$, and $t_{::k}$ is a frontal
slice.  See~\cite{KB09} for a survey and additional notation.

We will use the symbol $\otimes$ for the tensor (Kronecker, outer)
product: for two vectors $u \otimes v$ is a matrix with entries
$(u \otimes v)_{ij} = u(i)v(j)$, for three vectors $u \otimes v
\otimes w$ is a tensor with entries $(u \otimes v \otimes w)_{ijk}
= u(i)v(j)w(k)$. We say the tensor $u \otimes v \otimes w$ has
{\em rank $1$} unless it consists of zeros only, in which case it
has rank $0$. If a tensor $T$ can be written as a sum of at most
$r$ rank-$1$ tensors, we say $T$ has rank at most $r$. If $T = T_1
+ \cdots + T_r$, and each $T_i$ has rank at most $1$, we call
$(T_i)_{i=1}^r$ a {\em (rank-$r$) expansion} of $T$.

The following two results %and their proofs
are adapted from the conference version of H{\aa}stad's paper~\cite{H89}; in the journal version~\cite{H90} they were replaced by references to other papers.

%The original results were stated for fields, but we will see that it is easy to extend them to principal ideal domains (PID) such as \ZN.\footnote{A {\em principal ideal domain} is a commutative ring, with identity, without zero-divisors, and every ideal is generated by a single element. The main property we will use is that in a principal ideal domain any two elements have a greatest common divisor).

\begin{lemma}[H{\aa}stad~\cite{H89}]\label{lem:H1}
  Suppose $T = (t_{ijk})$ is a tensor of rank $r$ (over some field), and the slice $M = (t_{::k_1})$ has rank $1$, so $M = u_1 \otimes v_1$ for some $u_1, v_1$. Then $T$ can be written as $T = \sum_{\ell = 1}^r u_{\ell} \otimes v_{\ell} \otimes w_{\ell}$.
\end{lemma}

In other words, $T$ has a rank-$r$ expansion using the slice $M$ as one of the rank-$1$ terms.

\begin{comment}
\begin{proof}
Since $T$ has rank $r$, it can be written $T = \sum_{\ell = 1}^r A_{\ell} \otimes w_{\ell}$, where each $A_{\ell} = x_{\ell} \otimes y_{\ell}$
is a rank-$1$ matrix.
Since the $k = k_1$ slice of $T$ equals $M$, we have $M = \sum_{\ell=1}^r w_{\ell}(k_1) A_{\ell}$. Since $M$ has rank $1$, at least one of the $w_{\ell}(k_1)$ must be non-zero, and, without loss of generality, we can assume that $w_1(k_1) \neq 0$. If we define $A'_1 = M$ and $A'_{\ell} = A_{\ell}$ for $1 < \ell \leq r$, and let $w'_1 = w_1 / w_1(k_1)$ and
$w'_{\ell} = w_{\ell} - \frac{w_{\ell}(k_1)}{w_1(k_1)} w_1$, we have $T = \sum_{\ell = 1}^r A_{\ell} \otimes w_{\ell} = \sum_{\ell = 1}^r A'_{\ell} \otimes w'_{\ell}$. Since $A'_1 = M$, this is what we had to show.
\end{proof}
\end{comment}

\begin{lemma}[H{\aa}stad~\cite{H89}]\label{lem:H2}
 Suppose $T = \llb t_{i,j,k}\rrb$ is a tensor of rank $r$ (over some field), and there is a set of linearly independent slices $M_h = \llb t_{::h} \rrb$ of rank $1$, so $M_h = u_h \otimes v_h$, for $h \in H$. Then $T$ can be written as $T = \sum_{\ell = 1}^r u_{\ell} \otimes v_{\ell} \otimes w_{\ell}$.
\end{lemma}

In other words, if we have a set of linearly independent, rank-$1$ slices of a tensor, we can always assume that they occur in a minimum rank expansion of the tensor.

\subsection{Logic and Complexity}

Over a field (or ring) $\fd{F}$ we can define the existential theory $\ETh(\fd{F})$ of $\fd{F}$ as the set of all true existential first-order sentences in $\fd{F}$. We work over the signature $(0,1,+,*)$ and allow equality as predicate (for $\QN$ and $\RN$ we can define order from that: $x \geq 0$ if and only if $(\exists y_0,y_1,y_2,y_3) [x = y_0^2+y_1^2+y_2^2+y_3^2]$, using Lagrange's theorem for $\QN$).

\begin{lemma}[Buss, Frandsen, Shallit~\cite{BFS99}]\label{lem:BFS}
 Suppose $\fd{F}$ is a field (a commutative ring without zero divisors is sufficient). Given a first-order existential sentence over $\fd{F}$ one can construct (in polynomial time) a family of (multivariate) polynomials $p_1, \ldots, p_n$ with integer coefficients so that $\varphi$ is true if and only if $(\exists x)[p_1(x) = 0 \wedge \cdots \wedge p_n(x) = 0]$ is true over $\fd{F}$. If $\fd{F}$ is not algebraically closed, then we can assume that $n = 1$.
\end{lemma}

%\vskip0.1cm\hrule\vskip0.1cm

We write $\exists \fd{F}$ for the complexity class which is formed
by taking the polynomial-time downward closure of $\ETh(\fd{F})$.
Lemma~\ref{lem:BFS} then says that testing feasibility of a system
of polynomial equations over $\fd{F}$ is {\em complete} for the
complexity class $\exists \fd{F}$, that is, it is {\em hard} for
the complexity class (every problem in the class reduces to it),
and it lies in the class (feasibility of a polynomial system over
$\fd{F}$ can be tested in $\exists \fd{F}$).\footnote{In other
models, e.g., the Blum-Shub-Smale model~\cite{BSS89} this was
well-known earlier.} We are particularly interested in $\fd{F} \in
\{\GF_p, \QN, \RN, \CN\}$. We discussed relationships between
these complexity classes and traditional complexity classes in the
introduction.

%\vskip0.1cm\hrule\vskip0.1cm

 Since a polynomial with integer coefficients can be calculated
via a sequence of sums and products of variables and constants $1$ and $-1$, the following result follows
immediately from Lemma~\ref{lem:BFS}.

\begin{lemma}\label{lem:QE}
  Let $\fd{F}$ be a field (or commutative ring without zero divisors). Deciding whether a system of equations of the types $x_i = x_j + x_k$, $x_i = x_j x_k$, $x_i = x_j$, and $x_i = 1$, is solvable over $\fd{F}$ is complete for $\exists \fd{F}$.
\end{lemma}

Call a such system of equations a {\em quadratic system}.

Let us illustrate \NPR-completeness with an example relevant to
tensors. This corrects an example from Hillar and Lim~\cite[Remark
2.3]{HL13}.

\begin{example}[Hillar, Lim~\cite{HL13}]\label{ex:eig0}
A tensor $T = \llb t_{i,j,k} \rrb^{n,n,n}_{i,j,k=1}$ has {\em
eigenvalue} $\lambda$ if there is a non-zero vector $x$, the {\em
eigenvector}, so that
  \[ \sum^{n,n}_{i,j=1} t_{i,j,k} x_i x_j = \lambda x_k.\footnote{There are other definitions of eigenvalues for tensors as well.}\]
So $\lambda = 0$ is an eigenvalue of $T$ if there is a non-zero
vector $x$ satisfying $\sum^{n,n}_{i,j=1} t_{i,j,k} x_i x_j = 0$,
which is a homogenous quadratic system of equations, and,
obviously every homogenous quadratic system can be written in this
form. So deciding whether a tensor has $0$ as an eigenvalue is
computationally equivalent to deciding whether a homogenous
quadratic system has a non-trivial solution. This problem is
sometimes called \HTN\ (for Hilbert's homogenous Nullstellensatz),
and, over $\RN$, was shown to be \NPR-complete
in~\cite{S12}.\footnote{The proof in~\cite{S12} yields a quartic
systems, but that can be reduced to quadratic, by removing the
final (unnecessary) squaring operation.} Thus, deciding whether
$0$ is an eigenvalue of a tensor $T$ over $\RN$ is \NPR-complete.
Hillar and Lim~\cite[Remark 2.3]{HL13} also sketch a proof of
the \NPR-completeness of \HTN, but their proof of hardness of the quadratic homogenous
system is not correct; in their notation, they
require $z^2 = \sum_{i=1}^n x_i^2$, but this cannot be guaranteed.
For example, they would take the quadratic system $(x+2)^2 = 0$
and homogenize it as $x^2+4xz+4z^2 = 0$ and require $x^2 = z^2$. While the original system
has a non-trivial solution, $x = -2$, it is easy to see that the homogenized system only
has the trivial solution $x = z = 0$. The hardness proof seems to require a non-uniform construction as in~\cite{S12}.
%cannot be guaranteed; fixing this seems to require a construction
%like the one in~\cite{S12}, more precisely, one needs to use $z^2
%= w^2 + \sum_{i=1}^n x_i^2$ and an a priori bound on the
%$\sum_{i=1}^n x_i^2$ of a solution.
 \qed
\end{example}

\subsection{Algebraic Universality}

A solution to a system of algebraic equations may have high
complexity, e.g.,\ consider $x_0 = 1$, $x_1 = x_0 + x_0$, $x_2 =
x_1 x_1$, $\ldots$, $x_n = x_{n-1}x_{n-1}$. This system of $n+1$
equations defines a number $x_n$ requiring a bit expansion of
exponential length. Similarly, one can define a linear system
whose solution is an algebraic number of high degree.
\NPR-completeness reductions often preserve this property, so that
\NPR-complete problems require solutions of high complexity. For
example, Bienstock and Dean~\cite{B91, BD93} showed that any
straight-line drawing of a graph with the smallest number of
crossings may require vertex coordinates of double-exponential
precision. This is a very weak type of algebraic universality. A
stronger variant would, for example, show that for any algebraic
number there is a graph which contains that algebraic number
(after some normalization). A much stronger type of universality
result goes back to Mn\"{e}v~\cite{M88} who showed that any basis
semialgebraic set is homotopy (even stably) equivalent to the
realization space of a pseudoline arrangement. That is, for every
basic semialgebraic set Mn\"{e}v defines a pseudoline arrangement
so that the space of straight-line realizations of that pseudoline
arrangement is essentially the same as the basic semialgebraic set
up to some form of algebraic equivalence. We will show a weaker
type of algebraic universality for the tensor rank problem. To do
this properly, we need a definition of the realization space of a
rank-$r$ tensor. For a $3$-dimensional tensor $T \in \QN^{d_1
\times d_2 \times d_3}$, and integer $r$ define the rank-$r$
realization space of $T$ as
\[{\cal R}(T, r) := \{(u_1, v_1, w_1, \ldots, u_r, v_r, w_r): T = \sum_{\ell = 1}^r u_{\ell} \otimes v_{\ell} \otimes w_{\ell} \}.\]
Obviously, ${\cal R}(T,r) \subseteq \RN^{(d1+d2+d3)r}$ is an {\em algebraic set}; that is, it can be written as the set of common roots of a family of multivariate polynomials (with integer coefficients).

We would like to show that every algebraic set (with integer coefficients) over $\RN$ is essentially the same as some
${\cal R}(T,r)$ for some $T$ and $r$, but it seems to have too many degrees of freedom, so instead we work with
\begin{align*}
        {\cal R}(T, r, S, s) := & \{(w_1, \ldots, w_s): \exists (u_{s+1}, v_{s+1}, w_{s+1}, \ldots, u_r, v_r, w_r)  \\
                               & T = \sum_{\ell = 1}^s S_{\ell} \otimes w_{\ell}+  \sum_{\ell = s+1}^r u_{\ell} \otimes v_{\ell} \otimes w_{\ell}\}
\end{align*}
%\end{equation*}
where $S$ is a family of $s$ rank-$1$ matrices. In this version of ${\cal R}(T,r)$ we restrict the first $s$ products $u_{\ell} \otimes v_{\ell}$
to be $S_{\ell}$.

\begin{comment}
\begin{align*}
        {\cal R}(T, r, h) := \{& (u_1, v_1, w_1, \ldots, u_r, v_r, w_r): T = \sum_{\ell = 1}^r u_{\ell} \otimes v_{\ell} \otimes w_{\ell},\\
            &u_i \otimes v_i = (t_{::i})\ \mbox{for $1 \leq i \leq h$}\}.
\end{align*}
Note that ${\cal R}(T, r, h)$  is empty, if $t_{::i}$ does not have rank $1$ for all $1 \leq i \leq h$.
\end{comment}

We need to make precise the notion of being ``essentially the same'', we will use the notion of stable equivalence introduced by Richter-Gebert to uniformize various universality constructions~\cite{RG95,RG96}. Stable equivalence implies homotopy equivalence, and it maintains complexity of algebraic points~\cite{RG96}. Two sets are {\em rationally equivalent} if there
is a rational homeomorphism between the two sets.  A set $X$ is a {\em stable projection} of $Y$ if
\[Y = \{(y,y'): y \in X, \langle p_i(y), y'\rangle = c_i, 1 \leq i \leq n\}, \]
where the $p_i$ %and $q_i$
are multivariate polynomials with integer coefficients, and the $c_i$ are constants. Two sets are {\em stably equivalent} if they are in the same equivalence class with respect to stable projections and rational transformations.

We will show that for every algebraic set (with integer coefficients), there are $T$, $r$, $S$ and $s$ so that the algebraic set is stably equivalent to ${\cal R}(T, r, S, s)$, so this, restricted, tensor rank problem is universal for algebraic sets. By using ${\cal R}(T, r, S, s)$ instead of ${\cal R}(T,r)$ we side-step the fact that the two H{\aa}stad lemmas do not yield stable equivalence: forcing a particular $u_i \otimes v_i$ to equal a slice of $T$ changes the number of algebraic components of the solution set, so it cannot maintain homotopy equivalence.

\section{Hardness of Tensor Rank}

In this section we will see that the tensor rank problem over a field $\fd{F}$ is complete for $\exists \fd{F}$. In the Blum-Shub-Smale model, the same proof shows that the tensor rank problem over $\fd{F}$ is $\NP_\fd{F}$-complete. We will not discuss the Blum-Shub-Smale model in detail, and refer the reader to~\cite{BCSS98}.

\subsection{A Minimum Rank Problem}

For a matrix $A$ with entries being multinomials expressions in $\fd{F}[x_1, \ldots, x_n]$, the minrank of $M$ is the smallest (matrix) rank of $A$ over $\fd{F}$ achievable
by replacing variables $x_i$ with values in $\fd{F}$ and evaluating the resulting expressions.

\begin{definition}\label{pako0}
Let $\mrk_\fd{F}(A)$ be the minimum rank of $A$ (as a matrix over $\fd{F}$) over all possible assignments of
values in $\fd{F}$ to variables in $A$.
\end{definition}

Buss, Frandsen and Shallit~\cite{BFS99} showed that the minrank problem over $\fd{F}$ is complete for $\exists \fd{F}$, even if entries are restricted to be in $\fd{F} \cup \{x_1, \ldots, x_n\}$. We will show that the minrank problem is $\exists \fd{F}$-hard for matrices of a very specific form which lends itself to be turned into a tensor rank problem.\footnote{There is also a notion of minrank for matrices with entries in $\{+,-\}$. Given such a matrix is there a real matrix of rank at most $3$ with that sign pattern? This problem turns out to be \NPR-hard as well~\cite{BFGJK09,BK15}, but does not seem to be related to our minrank problem.}

Suppose we are given a quadratic system $S$ with $m$ equations $e_1,\dots,e_m$; we construct a square $3m\times 3m$ matrix $A$
with affine entries whose minrank will be connected to the feasibility of $S$ (see Definition~\ref{pako0} and
Lemma~\ref{pako1} below for a precise statement). To simplify the statements and the proofs we make the following
assumptions on the quadratic system:
\begin{description}
\item[A1] No variable occurs more than once in an equation.
\item[A2] Any two equations share at most one variable.
\item[A3] If $w=uv$ is an equation in $S$ then $v$ occurs exactly twice in $S$ and the other occurrence of $v$ is
in an equation of the form $v=z$.
\end{description}
Assumptions {\bf A1} and {\bf A2} are not restrictive since we can always ``copy'' a variable $v$ to a variable $v'$ using
equation $v'=v$ (and then use $v'$ in place of $v$). Assumption {\bf A3} is not restrictive since we can replace an
equation $w=uv$ by a pair of equations $v' = v,  w=uv'$, where $v'$ is a new variable.

The following $3\times 3$ matrices are the main building block in our construction
\begin{eqnarray}\label{A1}
\quad\quad \det\left(%
\begin{array}{ccc}
  1 & 0 & a \\
  0 & 1 & b \\
  1 & 1 & c \\
\end{array}%
\right) & = & c - (a+b), \\
\det\left(%
\begin{array}{ccc}
  1 & 0 & c \\
  0 & 1 & a \\
  -1 & b & 0 \\
\end{array}%
\right) & = & c - ab. \label{A2}
\end{eqnarray}
To construct the matrix $A$ we first place $3\times 3$
blocks on the diagonal as follows: The $\ell$-th diagonal $3\times 3$ block is given by
\begin{itemize}
\item the matrix in~\eqref{A1} if $e_\ell$ is of the form $c = a+b$,
\item the matrix in~\eqref{A2} if $e_\ell$ is of the form $c = a b$,
\item the matrix in~\eqref{A1} with $b=0$ if $e_\ell$ is of the form $c = a$,
\item the matrix in~\eqref{A1} with $b=0, a=K$ if $e_\ell$ is of the form $c = K$, where $K$ is a rational constant.
\end{itemize}
Note that equation $e_\ell$ is satisfied if and only if the determinant of the block is zero.
Let $R_u$ be the increasing list of rows that contain variable $u$ and let $C_u$ be the increasing
list of columns that contain variable~$u$. From assumption {\bf A1} it follows that a $3\times 3$
block contains at most one occurrence of $u$. Thus $|R_u|=|C_u|$ and
$u$ occurs at positions $(R_u[i],C_u[i])$ for $i=1\dots|R_u|$. Also note that for distinct variables $u,v$ we have that $R_u$ and
$R_v$ are disjoint (since in the matrices in~\eqref{A1} and~\eqref{A2} the
variables are in different rows). Now we add a few more entries into the matrix $A$. For every variable $u$, for every
$1\leq j\neq k\leq |R_u|$ we add an entry $u - u_j$, with new variable $u_j$, at position $(R_u[j],C_u[k])$ in $A$.
This completes the construction of matrix $A$.
\begin{observation}\label{bobo}
The construction satisfies the following:
\begin{enumerate}
\item $u$ occurs exactly at positions $R_u\times C_u$ and it always occurs with coefficient $1$,
\item the non-zero entries of $A$ outside of the diagonal $3\times 3$ blocks are at indices $\bigcup_u R_u\times C_u$,
\item $u_j$ only occurs in the $R_u[j]$-th row  and it always occurs with coefficient $-1$,
\item leaving out every 3rd row and every 3rd column of $A$ (that is, rows and columns whose index
is divisible by $3$) yields the $2m\times 2m$ identity matrix.
\end{enumerate}
\end{observation}
The third item in Observation~\ref{bobo} follows from assumption {\bf A3} and the form of the matrices in~\eqref{A1}
and~\eqref{A2}. Note that the only occurrence of a variable in a column whose index is not divisible by $3$
must come from ``$b$'' in~\eqref{A2}, that is, an equation of the form $c=ab$. The other occurrence of $b$ is
in a row whose index is divisible by $3$ (using assumption {\bf A3}). Since both occurrences of $b$ are in rows whose index is divisible by $3$
we have that $R_b\times C_b$ is in the left-out part of $A$. We showed that for every $u$ either all entries of $C_u$
or all entries of $R_u$ are divisible by $3$ and hence if we leave out every third column and every third row there
will be no off-diagonal entries.

We have the following connection between the quadratic system $S$ and its matrix $A$.
\begin{lemma}\label{pako1}
Assume that a quadratic system $S$ satisfies assumptions {\bf A1}, {\bf A2}, and {\bf A3}. Let $A$ be the matrix
corresponding to $S$. System $S$ is solvable over $\fd{F}$ if and only if $\mrk_\fd{F}(A) = 2m$.
\end{lemma}

\begin{example}\label{ex:sep}
Before proving Lemma~\ref{pako1} let us illustrate the construction with an example.
Let $S=\{u = x y, y = x, u = 2\}$.
Then the matrix $A$ corresponding to $S$ is
\begin{equation}\label{exot}
\left(%
\begin{array}{ccccccccc}
1 & 0 & u & 0 & 0 & 0 & 0 & 0 & u-u_1 \\
0 & 1 & x & 0 & 0 & x-x_1 & 0 & 0 & 0 \\
-1 & y & 0 & 0 & 0 & y-y_1 & 0 & 0 & 0 \\
0 & 0 & x-x_2 & 1 & 0 & x & 0 & 0 & 0 \\
0 & 0 & 0 & 0 & 1 & 0 & 0 & 0 & 0 \\
0 & y-y_2 & 0 & 1 & 1 & y & 0 & 0 & 0 \\
0 & 0 & 0 & 0 & 0 & 0 & 1 & 0 & 2 \\
\phantom{A}0\phantom{A} & \phantom{A}0\phantom{A} & \phantom{A}0\phantom{A} & \phantom{A}0\phantom{A} &
\phantom{A}0\phantom{A} & \phantom{A}0\phantom{A} & \phantom{A}0\phantom{A} & \phantom{A}1\phantom{A} & \phantom{A}0\phantom{A} \\
0 & 0 & u-u_2 & 0 & 0 & 0 & 1 & 1 & u
\end{array}%
\right)
\end{equation}
The quadratic system $S$ encodes the equation $x^2=2$. This equation has a solution over ${\mathbb R}$ and hence, by Lemma~\ref{pako1},
$\mrk_{\mathbb R}(A)=6$. On the other hand the equation does not have a solution over ${\mathbb Q}$ and hence,
by Lemma~\ref{pako1}, $\mrk_{\mathbb Q}(A)\geq 7$.
\qed\end{example}

\begin{proof}[Proof of Lemma~\ref{pako1}]
From Observation~\ref{bobo} (part 4) we have $\mrk_\fd{F}(A)\geq 2m$.

Suppose that $S$ has a solution $\sigma$ with values in $\fd{F}$. For each variable $u$ assign value $\sigma(u)$ to $u$ and
all $u_i$'s in $A$. Note that this assignment makes all entries outside the diagonal $3\times 3$ blocks
zero (since those entries are of the form $u-u_i$). Also note that each $3\times 3$ block has rank $2$ (since it contains a $2\times 2$
identity matrix and has determinant equal to zero---here we use the fact that $\sigma$ is a solution of $S$).
The rank of a block diagonal matrix is the sum of the ranks of the blocks and hence $\mrk_\fd{F}(A)= 2m$.

It remains to show that $\mrk_\fd{F}(A)= 2m$ implies that $S$ has a solution in $\fd{F}$. Let $\sigma$
be an assignment with values in $\fd{F}$ such that the rank of $\sigma(A)$ is $2m$. Consider the $\ell$-th $3\times 3$ diagonal
block $\hat{B}$. Let $\hat{A}$ be the matrix obtained from $\sigma(A)$ by leaving out every third row and every third column except for the column
and the row with index $3\ell$. Note that $\hat{A}$ is a $(2m+1)\times (2m+1)$ matrix and, by Observation~\ref{bobo} (part 4),
if we leave out the row and column with index $2\ell+1$ from $\hat{A}$ we get the identity matrix. Hence we have
\begin{eqnarray*}\label{enko}
\det(\hat{A}) & = &\hat{A}_{2\ell+1,2\ell+1} - \sum_{i\neq 2\ell+1} \hat{A}_{i,2\ell+1} \hat{A}_{2\ell+1,i} \\
                   & = & \det(\hat{B}) - \sum_{i\not\in\{2\ell-1,2\ell,2\ell+1\}} \hat{A}_{i,2\ell+1} \hat{A}_{2\ell+1,i}.
\end{eqnarray*}
We have
\begin{equation}\label{enko2}
\sum_{i\not\in\{2\ell-1,2\ell,2\ell+1\}} \hat{A}_{i,2\ell+1} \hat{A}_{2\ell+1,i} =  \sum_{i\neq\ell} A_{3i-2,3\ell} A_{3\ell,3i-2} + \sum_{i\neq\ell} A_{3i-1,3\ell} A_{3\ell,3i-1}.
\end{equation}
Note that $A_{3\ell,3i-2}=0$ for all $i\neq\ell$ since the first column in~\eqref{A1} and~\eqref{A2} does not contain any variables
(also see Observation~\ref{bobo} (part 2)). If $A_{3\ell,3i-1}\neq 0$ then the $i$-th block contains a variable in the $2$-nd column
(and hence in the $3$-rd row) and that variable also
occurs in the $3$-rd row of the $\ell$-th block. If $A_{3i-1,3\ell}\neq 0$ then the $i$-th block contains a variable in the $2$-nd row and that
variable also occurs in the $3$-rd column of the $\ell$-th block. Thus if both $A_{3\ell,3i-1}\neq 0$ and $A_{3i-1,3\ell}\neq 0$
then $e_i$ and $e_\ell$ would share two variables (occurring in the $2$-nd and $3$-rd row of the $i$-th block). This is
impossible (because of assumption {\bf A2}) and hence equation~\eqref{enko2} has value $0$.
We conclude that $\det(\hat{A}) = \det(\hat{B})$.

Now $\hat{A}$ has rank at most $2m$, since $\sigma(A)$ has rank $2m$, but dimension $(2m+1)\times (2m+1)$, so its columns are linearly dependent, and we conclude that
$$
0 = \det(\hat{A}) = \det(\hat{B})
$$
and hence the $\ell$-th equation is satisfied by the assignment $\sigma$, for all $\ell\in [m]$. Thus $\sigma$ is a
solution of $S$ in~$\fd{F}$.
\end{proof}

\subsection{A Tensor Rank Problem}

We are left with translating the minrank problem from the previous section into a tensor rank problem. Recall that given
a quadratic system $S$ we constructed a matrix $A$ consisting of diagonal blocks (with constants and variable terms) and additional, affine entries in rows and columns divisible by $3$.

Define a tensor $T_A$ from $A$ as follows:
\begin{itemize}
  \item for every variable $x$ in $A$ let the partial derivative $A_x := \partial A/\partial x$ be a (frontal) slice of $T$; $\partial A/\partial x$  is the matrix containing the coefficients of $x$ in $A$,
  \item add one final, that is, $(n+1)$-st (frontal) slice $A_1$ containing all the constant values of $A$.
\end{itemize}

Note that if $\sigma$ assigns a value in $\fd{F}$ to each variable in $A$, then $\sigma(A) = A_1 + \sum_{x} \sigma(x) A_x$.
Let $n$ be the number of variables in $A$; $T_A$ is a $3m \times 3m \times (n+1)$ tensor.

\begin{lemma}\label{lem:trans}
 $A$ has minrank at most $2m$ if and only if $T_A$ has tensor rank at most $2m+n$.
\end{lemma}
\begin{proof}
 If $A$ has minrank $2m$, then there is a $\sigma$ assigning $\sigma(x) \in \fd{F}$ to each variable $x$ occurring in $A$
 so that the rank of $\sigma(A)$ is $2m$. Now $\sigma(A) = A_1 + \sum_x \sigma(x) A_x$, where the sum is over all $n$ variables $x$ occurring in $A$.
 In other words, $A_1 = \sigma(A) + \sum_x (-A_x)$. Since $\sigma(A)$ has matrix rank $2m$, it can be written as the sum of $2m$ rank-$1$ matrices, so
 $A_1$ can be written as the sum of $2m+n$ rank-$1$ matrices---each $A_x$ has rank $1$. Hence, every slice of $T_A$ can be written using the $A_x$ and the $2m$ rank-$1$ matrices summing up to $A_1$, implying that $T_A$ has tensor rank at most $2m+n$.

 For the other direction, assume that $T_A$ has tensor rank at most $2m+n$.
 We first observe that the $n$ matrices $A_x$ are linearly independent: Suppose that
 $\sum_x \lambda(x) A_x = 0$ for some vector $\lambda$. The matrix $A$ contains two types of variables: the original
 variables $u$ (from the quadratic system), and the additional variables $u_j$. Now any non-zero entry in $A_u$ is unique in the sense that no other $A_x$ has an entry in the same position, so $\lambda(u) = 0$ for the original variables. But then any non-zero entry in $A_{u_j}$ is unique among the remaining matrices (belonging to the non-original variables), so $\lambda(u_j) = 0$ for all remaining variables, establishing $\lambda = 0$. Therefore, the $A_x$ are linearly independent.

 Lemma~\ref{lem:H2} now implies that $T_A$ can be written using the $A_x$ and $2m$ additional rank-$1$ tensors. So
 $T_A = \sum_x A_x \otimes z_x + \sum_{i=1}^{2m} u_i \otimes v_i \otimes w_i$, and,
 in particular, looking at the $n+1$st frontal slice of $T_A$, which is $A_1$, we obtain
 \[A_1 = \sum_{x} \tau(x) A_x + \sum_{i=1}^{2m} B_i,\]
 where $\tau(x) = z_x[n+1]$, and $B_i = w_i[n+1] (u_i \otimes v_i)$, where the $B_i$ are rank-$1$ matrices. In other words, $A_1 - \sum_{x} \tau(x) A_x = \sum_{i=1}^{2m} B_i$ has matrix rank at most $2m$.
 Setting $\sigma(x) := - \tau(x)$ we have that $A_1 + \sum_{x} \sigma(x) A_x$ has rank $2m$, and, moreover, equals $\sigma(A)$. But this shows that the minrank of $A$ is at most $2m$, which is what we had to prove.
\end{proof}

The following is a well-known result. For more results on tensor rank over various rings, see Howell~\cite{H78}.

\begin{corollary}\label{ex:sep}
 There is a tensor $T$ with $\rk_{\QN}(T) > \rk_{\RN}(T)$.
\end{corollary}
\begin{proof}
 Let $A$ be the matrix from Example~\ref{ex:sep}, and consider the tensor $T_A$ constructed in Lemma~\ref{lem:trans}. Then $\rk_{\QN}(T_A) \geq 7 + 9= 16$, while $\rk_{\RN}(T_A) = 6+9 = 15$.
\end{proof}

We can now complete the proof of our main result.

\begin{proof}[Proof of Theorem~\ref{thm:trans}]
Lemmas~\ref{lem:BFS} and~\ref{lem:QE} allow us to translate $\varphi$ into a quadratic system $S$ so that $\varphi$ is true over $\fd{F}$ if and only if $S$ has a solution over $\fd{F}$. Lemma~\ref{pako1} translates $S$ into a minrank problem over a matrix $A$, and Lemma~\ref{lem:trans} turns that into a tensor rank problem over $\fd{F}$.
\end{proof}

\subsection{Universality}

Reviewing the hardness proofs carefully shows that they also yield algebraic universality. Let us start with the minrank problem:

\begin{corollary}\label{cor:minrankuni}
    For every algebraic set $V$ specified using integer coefficients, we can find a matrix $A$ whose entries are multilinear expressions in $\fd{F}[x_1, \ldots, x_m]$, and an integer $k$ so that $V$ is stably equivalent to $\{(x_1, \ldots, x_d): \mrk_\fd{F}(A) = k\}$.
\end{corollary}
\begin{proof}
    Suppose we are given an algebraic set $V = \{(x_1, \ldots, x_d) \in \fd{F}^d: p_1(x_1, \ldots, x_d) = \cdots = p_n(x_1, \ldots, x_d) = 0\}$.
    We transform the system $p_1(x_1, \ldots, x_d) = \cdots p_n(x_1, \ldots, x_d) = 0$ into a quadratic system $S$ (as in Lemma~\ref{lem:QE}). While $S$ may require additional variables, each of these is equal to a polynomial transformation of the $x_i$ so that the realization space of $S$ is stably equivalent to the original algebraic set $V$ (in this case via a rational transformation).
    In the next step, we turn
    $S$ into a matrix $A$ with multilinear expressions over $x_1, \ldots, x_m$, and an integer $k$ as in Lemma~\ref{pako1} so that $S$ is solvable if and only if $\mrk_\fd{F}(A) = k$. Moreover, the variables of $S$ are variables of $A$, though $A$ may contain additional variables. However, those, as before, equal existing variables when $\mrk_\fd{F}(A) = k$, so $S$ is stably equivalent to
     $\{(x_1, \ldots, x_d): \mrk_\fd{F}(A) = k\}$, and then, by transitivity, so is $V$.
\end{proof}

In other words, the minrank problem for matrices with multilinear expressions over a field is universal for algebraic sets over that field. This gives us universality of the tensor problem as well.

\begin{corollary}\label{cor:tensoruni}
 For every algebraic set $V$ we can find a tensor $T$, an integer $r$, and a family of $s$ rank-$1$ matrices $S$ %and a family of $r$ rank-$1$ matrices $S$
 so that $V$ is stably equivalent to the realization space ${\cal R}(T, r, S, s)$.
\end{corollary}
\begin{proof}
    By Corollary~\ref{cor:minrankuni}, the algebraic set $V$ is stably equivalent to a minrank problem $\mrk_\fd{F}(A) = k_A$ for matrix $A$ and $k_A$ as constructed in the proof of Lemma~\ref{pako1}. From $A$ we construct a $3m \times 3m \times (n+1)$ tensor $T$ and an integer $k = 2m+n$, as in Lemma~\ref{lem:trans}, so that $V \neq \emptyset$ if and only if the tensor rank of $T$ is at most $k$. We know what the potential basis for $T$ looks like: it consists of the $n$ matrices $A_{x_i}$, the coefficient matrix of $x_i$, and $2m$ matrices $B_i$, two for each of the $m$ blocks in the minrank problem (keeping first and second column in each block). As in the proof of Lemma~\ref{lem:trans} we can argue that the $A_{x_i}$
    occur in the basis, since they are linearly independent. Define $S_i = A_{x_i}$, for $1 \leq i \leq n$.

    Consider an element of the realization space
    \begin{align*}
       {\cal R}(T, k, S, n) := & \{(w_1, \ldots, w_n): \exists (u_{n+1}, v_{n+1}, w_{n+1}, \ldots, u_k, v_k, w_k)  \\
                               & T = \sum_{\ell = 1}^n S_{\ell} \otimes w_{\ell}+  \sum_{\ell = n+1}^k u_{\ell} \otimes v_{\ell} \otimes w_{\ell}\},
    \end{align*}
    where $k = 2m+n$.
    Recall that $(t_{::n+1}) = A_1$, the matrix of constants from the minrank problem, so
    \begin{equation}\label{eq:A1}
    A_1 = \sum_{i = 1}^n w_i[n+1] A_{x_i} + \sum_{i=n+1}^{k} w_i[n+1] (u_i \otimes v_i).
    \end{equation}
    %which means $A_1 = \sum_{i = 1}^n w_i[n+1] A_{x_i} + \sum_{i=n+1}^r S_i$.
    This implies, as we argued in Lemma~\ref{lem:trans}, that $-w_i[n+1]$, for $1 \leq i \leq n$, is the value of $x_i$ in a solution $x = (x_1, \ldots, x_n)$ of the minrank problem. To prove stable equivalence, we have to show that the remaining $n$ slices of $w_i$ can be determined as well. As the second claim below shows they are even constant.
    We claim that
    \begin{itemize}
     \itemi $w_i[\ell] = 0$ for all $1 \leq i \leq n$ and $n+1 \leq \ell \leq k$; and
     \itemii $w_i[\ell] = \delta_{i\ell}$ for $1 \leq i, \ell \leq n$, where $\delta_{i\ell}$ is the Kronecker $\delta$.
     \end{itemize}
    To see claim $(i)$, we rewrite Equation~\eqref{eq:A1} as
    \[ A_1 - \sum_{i = 1}^n w_i[n+1] A_{x_i} = \sum_{i=n+1}^{k} w_i[n+1] (u_i \otimes v_i).\]
    Dropping every third row and column leaves us with the $2m \times 2m$ identity matrix on the left-hand side, so
     $u_{n+1}\otimes v_{n+1}, \dots, u_{k}\otimes v_{k}$
     must be linearly independent. The $\ell$-th slice of $T$, for $1 \leq \ell \leq n$, is
     \begin{equation}\label{eq:Axell}
      A_{x_{\ell}} = \sum_{i = 1}^n w_i[\ell] A_{x_i} + \sum_{i=n+1}^{k} w_i[\ell] (u_i \otimes v_i),
      \end{equation}
    Rewriting as before
     \[ A_{x_{\ell}} - \sum_{i = 1}^n w_i[\ell] A_{x_i} = \sum_{i=n+1}^{k} w_i[\ell] (u_i \otimes v_i),\]
     and again dropping every third row and column leaves
    us with the null matrix on the left-hand side, which, by independence of the $u_i \otimes v_i$, implies that $w_i[\ell] = 0$ for
    all $n+1 \leq \ell \leq k$, proving $(i)$.

    Claim $(ii)$ now follows by using claim $(i)$ in Equation~\eqref{eq:Axell} to obtain that
     \[ A_{x_{\ell}} = \sum_{i = 1}^n w_i[\ell] A_{x_i}.\]
     Since the $A_{x_i}$ are independent (as we argued in the proof of Lemma~\ref{lem:trans}),
     this implies that $w_i[\ell] = \delta_{i\ell}$ for all $1 \leq i \leq n$ and $1\leq \ell \leq n$.

    We conclude that ${\cal R}(T, k, S, n)$ is stably equivalent to the minrank problem, and, thus, to $V$.
 \end{proof}

\section{Open Questions}

There are several natural follow-up questions suggested by the results of this paper. For example, what is the complexity of tensor rank for symmetric tensors? Is tensor-rank hard for a fixed rank ($2$ or $3$ even) or is it fixed-parameter tractable? Over the complex numbers, Koiran's result places the problem at the second level of the polynomial hierarchy assuming the Generalized Riemann hypothesis is true. With the recent successes of exact algorithms for \NP-hard problems, is there a way to make Koiran's result algorithmic? Is there a way to remove the assumption?

\begin{comment}
Finally, what is the complexity of the tensor rank over $\ZN$ (or PIDs)? In spite of several attempts we were not able to push the construction to cover this case, indeed, even H{\aa}stad's two lemmas are problematic over $\ZN$. We conjecture that the problem is undecidable.
\end{comment}

\begin{comment}
\begin{remark}
  Is it true that $\rk_{\ZN}(M) = \max_{p} \rk_{\GF_p}(M)$ for all $M$? That would simplify the proof of undecidability of tensor rank over $\ZN$, since we could apply this result in the essential step. The answer is: probably no, separating polynomials such as in Polynomials With Roots Modulo Every Integer (Daniel Berend And Yuri Bilu), can probably be used to separate rank over integers from rank over finite fields.
\end{remark}
\end{comment}

\begin{comment}
symmetric tensor; hyperdeterminant?
fixed tensor rank: complexity? Does tensor rank 1 or 2 encode the sum of square roots problem?
- rank 1 is easy: all frontal slices are factors of each other, rank 2?
- can Koiran be made algorithmic? can hypothesis be removed
- fixed tensor rank decidable for Q or Z? in general: complexity?
- fixed dimension decidable for Q or Z? in general: complexity?
spectral norm (for matrices)?

integers?
\end{comment}

\bibliographystyle{plain}
\bibliography{tensor}

\begin{thebibliography}{10}

\bibitem{ABKDM06}
Eric Allender, Peter Burgisser, Johan Kjeldgaard-Pedersen, and Peter~Bro
  Miltersen.
\newblock On the complexity of numerical analysis.
\newblock In {\em CCC '06: Proceedings of the 21st Annual IEEE Conference on
  Computational Complexity}, pages 331--339, Washington, DC, USA, 2006. IEEE
  Computer Society.

\bibitem{BFGJK09}
Ronen Basri, Pedro~F. Felzenszwalb, Ross~B. Girshick, David~W. Jacobs, and
  Caroline~J. Klivans.
\newblock Visibility constraints on features of 3{D} objects.
\newblock In {\em CVPR}, pages 1231--1238. IEEE Computer Society, 2009.

\bibitem{BK15}
Amey Bhangale and Swastik Kopparty.
\newblock The complexity of computing the minimum rank of a sign pattern
  matrix.
\newblock {\em CoRR}, abs/1503.04486, 2015.

\bibitem{B91}
Daniel Bienstock.
\newblock Some provably hard crossing number problems.
\newblock {\em Discrete Comput. Geom.}, 6(5):443--459, 1991.

\bibitem{BD93}
Daniel Bienstock and Nathaniel Dean.
\newblock Bounds for rectilinear crossing numbers.
\newblock {\em J. Graph Theory}, 17(3):333--348, 1993.

\bibitem{BCSS98}
Lenore Blum, Felipe Cucker, Michael Shub, and Steve Smale.
\newblock {\em Complexity and real computation}.
\newblock Springer-Verlag, New York, 1998.

\bibitem{BSS89}
Lenore Blum, Mike Shub, and Steve Smale.
\newblock On a theory of computation and complexity over the real numbers:
  {NP}-completeness, recursive functions and universal machines.
\newblock {\em Bull. Amer. Math. Soc. (N.S.)}, 21(1):1--46, 1989.

\bibitem{BFS99}
Jonathan~F. Buss, Gudmund~S. Frandsen, and Jeffrey~O. Shallit.
\newblock The computational complexity of some problems of linear algebra.
\newblock {\em J. Comput. System Sci.}, 58(3):572--596, 1999.

\bibitem{C88}
John Canny.
\newblock Some algebraic and geometric computations in pspace.
\newblock In {\em STOC '88: Proceedings of the twentieth annual ACM symposium
  on Theory of computing}, pages 460--469, New York, NY, USA, 1988. ACM.

\bibitem{DMR76}
Martin Davis, Yuri Matijasevi{\v{c}}, and Julia Robinson.
\newblock Hilbert's tenth problem: {D}iophantine equations: positive aspects of
  a negative solution.
\newblock In {\em Mathematical developments arising from {H}ilbert problems
  ({P}roc. {S}ympos. {P}ure {M}ath., {V}ol. {XXVIII}, {N}orthern {I}llinois
  {U}niv., {D}e {K}alb, {I}ll., 1974)}, pages 323--378. (loose erratum). Amer.
  Math. Soc., Providence, R. I., 1976.

\bibitem{H89}
Johan H{\aa}stad.
\newblock Tensor rank is {NP}-complete.
\newblock In {\em Automata, languages and programming ({S}tresa, 1989)}, volume
  372 of {\em Lecture Notes in Comput. Sci.}, pages 451--460. Springer, Berlin,
  1989.

\bibitem{H90}
Johan H{\aa}stad.
\newblock Tensor rank is {NP}-complete.
\newblock {\em J. Algorithms}, 11(4):644--654, 1990.

\bibitem{HL13}
Christopher~J. Hillar and Lek-Heng Lim.
\newblock Most tensor problems are {NP}-hard.
\newblock {\em J. ACM}, 60(6):Art. 45, 39, 2013.

\bibitem{H78}
Thomas~D. Howell.
\newblock Global properties of tensor rank.
\newblock {\em Linear Algebra Appl.}, 22:9--23, 1978.

\bibitem{K10}
Jochen {Koenigsmann}.
\newblock Defining $\mathbb{Z}$ in $\mathbb{Q}$.
\newblock {\em ArXiv e-prints}, 2010.

\bibitem{K96}
Pascal Koiran.
\newblock Hilbert's {N}ullstellensatz is in the polynomial hierarchy.
\newblock {\em J. Complexity}, 12(4):273--286, 1996.
\newblock Special issue for the Foundations of Computational Mathematics
  Conference (Rio de Janeiro, 1997).

\bibitem{KB09}
Tamara~G. Kolda and Brett~W. Bader.
\newblock Tensor decompositions and applications.
\newblock {\em SIAM Rev.}, 51(3):455--500, 2009.

\bibitem{KM94}
Jan Kratochv{\'{\i}}l and {Ji{\v{r}}\'\i} Matou{\v{s}}ek.
\newblock Intersection graphs of segments.
\newblock {\em J. Combin. Theory Ser. B}, 62(2):289--315, 1994.

\bibitem{M70}
Ju.~V. Matijasevi{\v{c}}.
\newblock The {D}iophantineness of enumerable sets.
\newblock {\em Dokl. Akad. Nauk SSSR}, 191:279--282, 1970.

\bibitem{M88}
N.~E. Mn{\"e}v.
\newblock The universality theorems on the classification problem of
  configuration varieties and convex polytopes varieties.
\newblock In {\em Topology and geometry---{R}ohlin {S}eminar}, volume 1346 of
  {\em Lecture Notes in Math.}, pages 527--543. Springer, Berlin, 1988.

\bibitem{P09}
Bjorn Poonen.
\newblock Characterizing integers among rational numbers with a
  universal-existential formula.
\newblock {\em Amer. J. Math.}, 131(3):675--682, 2009.

\bibitem{RG95}
J{\"u}rgen Richter-Gebert.
\newblock Mn\"ev's universality theorem revisited.
\newblock {\em S\'em. Lothar. Combin.}, 34, 1995.

\bibitem{RG96}
J{\"u}rgen Richter-Gebert.
\newblock {\em Realization spaces of polytopes}, volume 1643 of {\em Lecture
  Notes in Mathematics}.
\newblock Springer-Verlag, Berlin, 1996.

\bibitem{S09b}
Marcus Schaefer.
\newblock Complexity of some geometric and topological problems.
\newblock In David Eppstein and Emden~R. Gansner, editors, {\em Graph Drawing},
  volume 5849 of {\em Lecture Notes in Computer Science}, pages 334--344.
  Springer, 2009.

\bibitem{S12}
Marcus Schaefer.
\newblock Realizability of graphs and linkages.
\newblock In J\'{a}nos Pach, editor, {\em Thirty Essays on Geometric Graph
  Theory}, pages 461--482. Springer, 2012.

\bibitem{SS15}
Marcus Schaefer and Daniel {\v{S}}tefankovi{\v{c}}.
\newblock Fixed points, {N}ash equilibria, and the existential theory of the
  reals.
\newblock {\em Theory of Computing Systems}, pages 1--22, 2015.

\bibitem{S16}
Yaroslav Shitov.
\newblock How hard is the tensor rank?
\newblock {\em CoRR}, abs/1611.01559, 2016.

\bibitem{S91}
Peter~W. Shor.
\newblock Stretchability of pseudolines is {NP}-hard.
\newblock In {\em Applied geometry and discrete mathematics}, volume~4 of {\em
  DIMACS Ser. Discrete Math. Theoret. Comput. Sci.}, pages 531--554. Amer.
  Math. Soc., Providence, RI, 1991.

\end{thebibliography}

\end{document}